\documentclass[reqno]{amsart}
\usepackage{fullpage,amsfonts,amssymb,amsthm,mathrsfs,graphicx,color,framed}
\usepackage{cancel,bm}

\usepackage[colorlinks=true, pdfstartview=FitV, linkcolor=blue, citecolor=blue, urlcolor=blue]{hyperref}
\definecolor{shadecolor}{rgb}{0.95, 0.95, 0.86}

\renewcommand{\d}{{\mathrm d}}

\newcommand{\im}{\mathrm{i}}
\newcommand{\e}{\mathrm{e}}

\numberwithin{equation}{section}

\newtheorem{theo}{Theorem}[section]

\newtheorem{rem}[theo]{Remark}

\newtheorem{prop}[theo]{Proposition} 
\newtheorem{cor}[theo]{Corollary}

\begin{document}

\title[Short distance asymptotics]{Short distance asymptotics for a generalized two-point scaling function in the two-dimensional Ising model}

\author{Thomas Bothner}
\address{Department of Mathematics, University of Michigan, 2074 East Hall, 530 Church Street, Ann Arbor, MI 48109-1043, United States}
\email{bothner@umich.edu}

\author{William Warner}
\address{Department of Mathematics, University of Michigan, 2074 East Hall, 530 Church Street, Ann Arbor, MI 48109-1043, United States}
\email{wipawa@umich.edu}



\thanks{The results of this article grew out of an eight week long Research Experience for Undergraduates (REU) program hosted at the University of Michigan in summer 2018. The work of T.B. is supported by the AMS and the Simons Foundation through a travel grant and W.W. acknowledges financial support provided by the Michigan Center for Applied and Interdisciplinary Mathematics. Both authors are grateful to C. Doering for stimulating discussions.}

\begin{abstract} In the 1977 paper \cite{MTW} of B. McCoy, C. Tracy and T. Wu it was shown that the limiting two-point correlation function in the two-dimensional Ising model is related to a second order nonlinear Painlev\'e function. This result identified the scaling function as a tau-function and the corresponding connection problem was solved by C. Tracy in 1991 \cite{T}, see also the works by C. Tracy and H. Widom in 1998 \cite{TW}. Here we present the solution to a certain generalized version of the above connection problem which is obtained through a refinement of the techniques chosen in \cite{B}.
\end{abstract}

\date{\today}
\maketitle
\section{Introduction and statement of results}\label{sec:11}
This note is concerned with the solution of a generalized connection problem for a distinguished tau-function of the $\nu$-modified radial sinh-Gordon equation.
\subsection{Modified sinh-Gordon equation and connection problem}
In 1977, B. McCoy, C. Tracy and T. Wu derived the following result.
\begin{theo}[McCoy, Tracy, Wu \cite{MTW}, 1977] Let
\begin{equation*}
	f_{2n}(t;\nu):=\frac{(-1)^n}{n}\int_1^{\infty}\cdots\int_1^{\infty}\left[\prod_{j=1}^{2n}\frac{\e^{-ty_j}}{\sqrt{y_j^2-1}}\left(\frac{y_j-1}{y_j+1}\right)^{\nu}\frac{1}{y_j+y_{j+1}}\right]\prod_{j=1}^n\big(y_{2j}^2-1\big)\prod_{j=1}^{2n}\d y_j,\ \ n\in\mathbb{Z}_{\geq 1},
\end{equation*}
with $y_{2n+1}\equiv y_1,t>0$ and $\nu>-\frac{1}{2}$. Then for any $\lambda\in[0,\frac{1}{\pi}]$,
\begin{equation}\label{e:1}
	\exp\left[-\sum_{n=1}^{\infty}\lambda^{2n}f_{2n}(t;\nu)\right]=\exp\bigg\{\frac{1}{4}\int_t^{\infty}\left[\sinh^2\psi-\left(\frac{\d\psi}{\d s}\right)^2+\frac{4\nu}{s}\sinh^2\left(\frac{\psi}{2}\right)\right]s\,\d s\bigg\}\cosh\left(\frac{\psi}{2}\right),
\end{equation}
where $\psi=\psi(t;\nu,\lambda)\in\mathbb{R}$ solves the second order nonlinear ODE
\begin{equation}\label{e:2}
	\frac{\d^2\psi}{\d t^2}+\frac{1}{t}\frac{\d\psi}{\d t}=\frac{1}{2}\sinh(2\psi)+\frac{2\nu}{t}\sinh\psi,
\end{equation}
subject to the boundary condition
\begin{equation}\label{Watson}
	\psi(t;\nu,\lambda)\sim 2\lambda\int_1^{\infty}\frac{\e^{-ty}}{\sqrt{y^2-1}}\left(\frac{y-1}{y+1}\right)^{\nu}\,\d y\ \ \ \ \textnormal{as}\ \ t\rightarrow+\infty.
\end{equation}
\end{theo}
The right-hand side in \eqref{e:1}, which we shall abbreviate as $\tau(t;\nu,\lambda)$ below, appeared first in the Wu, McCoy, Tracy, Barouch analysis \cite{BMTW} of the scaling limit of the spin-spin correlation function in the 2D Ising model. In a nutshell (cf. \cite{MW,P} for more details), if $\xi=\xi(T)$ is the temperature dependent correlation length (which diverges as $|T-T_c|^{-1}$ near the critical temperature $T_c$) and $\langle\sigma_{00}\sigma_{MN}\rangle$ the 2-point function on an isotropic Onsager lattice after the thermodynamic limit, then
\begin{equation}\label{e:3}
	R^{\frac{1}{4}}\langle\sigma_{00}\sigma_{MN}\rangle\rightarrow F_{\pm}(t),
\end{equation}
holds true in the {\it massive scaling limit} $\xi\rightarrow\infty,R=\sqrt{M^2+N^2}\rightarrow\infty$ such that $t=R/\xi>0$ is fixed. The $\pm$ choice refers to the scaling limit taken either above or below the critical temperature $T_c$ and $F_{\pm}(t)$ are the so-called {\it scaling functions} given by
\begin{equation}\label{scaling}
	F_-(t)=2^{\frac{3}{8}}t^{\frac{1}{4}}\tau\left(t,0,\frac{1}{\pi}\right)\ \ \ \ \ \  \ \ \textnormal{and}\ \ \ \ \ \ \ \ F_+(t)=F_-(t)\tanh\left[\frac{1}{2}\psi\left(t;0,\frac{1}{\pi}\right)\right].
\end{equation}
The nonlinear equation \eqref{e:2}, coined $\nu$-modified radial sinh-Gordon equation, is intimately related to Painlev\'e special function theory since $u(t;\nu,\lambda):=\e^{-\psi(2t;\nu,\lambda)}$ solves
\begin{equation}\label{e:4}
	\frac{\d^2 u}{\d t^2}=\frac{1}{u}\left(\frac{\d u}{\d t}\right)^2-\frac{1}{t}\frac{\d u}{\d t}+\frac{2\nu}{t}\big(u^2-1\big)+u^3-\frac{1}{u},
\end{equation}
which, cf. \cite[$32.2.3$]{NIST}, is Painlev\'e-III with constants $(\alpha,\beta,\gamma,\delta)=(2\nu,-2\nu,1,-1)$. Historically, \eqref{e:2}, resp. \eqref{e:4}, was the very first appearance of a Painlev\'e function in a problem of mathematical physics - predating various field theoretic or nonlinear wave theoretic applications and in particular predating their numerous appearances in quantum gravity, enumerative toplogy, random matrix theory, combinatorics, integrable probability, etc. It is known that $\psi(t;\nu,\lambda)$, uniquely determined by the above specified $+\infty$ behavior, is a highly transcendental function which cannot be expressed in terms of known classical special functions (i.e. not in terms of a finite number of contour integrals of elementary, or elliptic or finite genus algebraic functions). This fact turns the underlying connection problem, i.e. the problem of computing the $t\downarrow 0$ behavior of $\psi(t;\nu,\lambda)$ from its known $t\rightarrow+\infty$ asymptotics (or vice versa), into a challenging and foremost nonstandard task. Nowadays powerful analytical techniques of inverse scattering, isomonodromy or Riemann-Hilbert theory allow us to derive the underlying connection formul\ae\, in a systematic fashion. Again, \cite{MTW} predates these approaches and McCoy, Tracy and Wu had, for the very first time, derived a complete connection formula in case of \eqref{e:2}, resp. \eqref{e:4}. In detail, they showed in \cite[$1.10$]{MTW} that for fixed $\lambda\in[0,\frac{1}{\pi})$ the small $t$-expansion of $u(t;\nu,\lambda)$ is given by
\begin{equation}\label{e:5}
	u\Big(\frac{t}{2};\nu,\lambda\Big)=Bt^{\sigma}\left(1-\frac{\nu}{B}(1-\sigma)^{-2}t^{1-\sigma}+B\nu(1+\sigma)^{-2}t^{1+\sigma}+\mathcal{O}\big(t^{2(1-\sigma)}\big)\right),
\end{equation}
with $(t,\nu,\lambda)$-differentiable error terms and where $(\sigma,B)$ are the following functions of $(\nu,\lambda)$,
\begin{equation}\label{e:6}
	\sigma=\frac{2}{\pi}\arcsin(\pi\lambda)\in[0,1);\ \ \ \ B=2^{-3\sigma}\frac{\Gamma^2(\frac{1}{2}(1-\sigma))}{\Gamma^2(\frac{1}{2}(1+\sigma))}\frac{\Gamma(\nu+\frac{1}{2}(1+\sigma))}{\Gamma(\nu+\frac{1}{2}(1-\sigma))},
\end{equation}
in terms of Euler's Gamma function $\Gamma(z)$, cf. \cite[$5.2.1$]{NIST}. Note that $B=B(\nu,\sigma)$ might vanish for $\nu<0$, thus expansion \eqref{e:5} holds uniformly in $(\sigma,\nu)\in[0,1)\times(-\frac{1}{2},+\infty)\setminus\{\sigma=1+2\nu\}$ chosen from compact subsets. Also, $B=B(\nu,\sigma)$ is negative iff $\sigma>1+2\nu$, compare Figure \ref{figure1} below. Hence, given our connection of  $u(t;\nu,\lambda)=\e^{-\psi(2t;\nu,\lambda)}\geq 0$ to the real-valued scaling functions in the 2D Ising model \eqref{e:1}, we shall in the following restrict ourselves, whenever working with \eqref{e:5} and \eqref{e:6}, to the values $(\sigma,\nu)\in[0,1)\times(-\frac{1}{2},+\infty)$ for which $\sigma<1+2\nu$.
\begin{center}
\begin{figure}[tbh]
\resizebox{0.4515\textwidth}{!}{\includegraphics{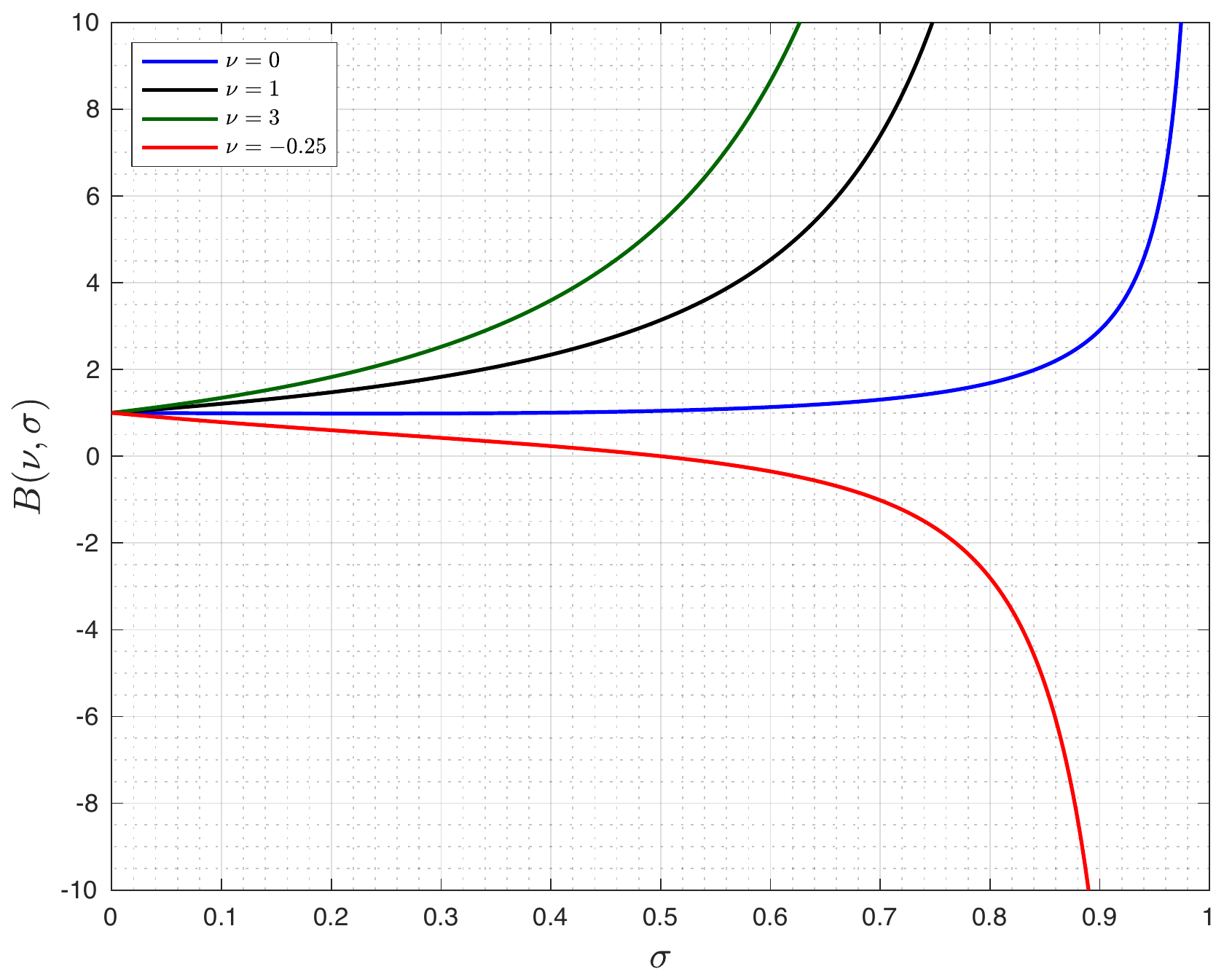}}
\ \ \ \ \ \resizebox{0.447\textwidth}{!}{\includegraphics{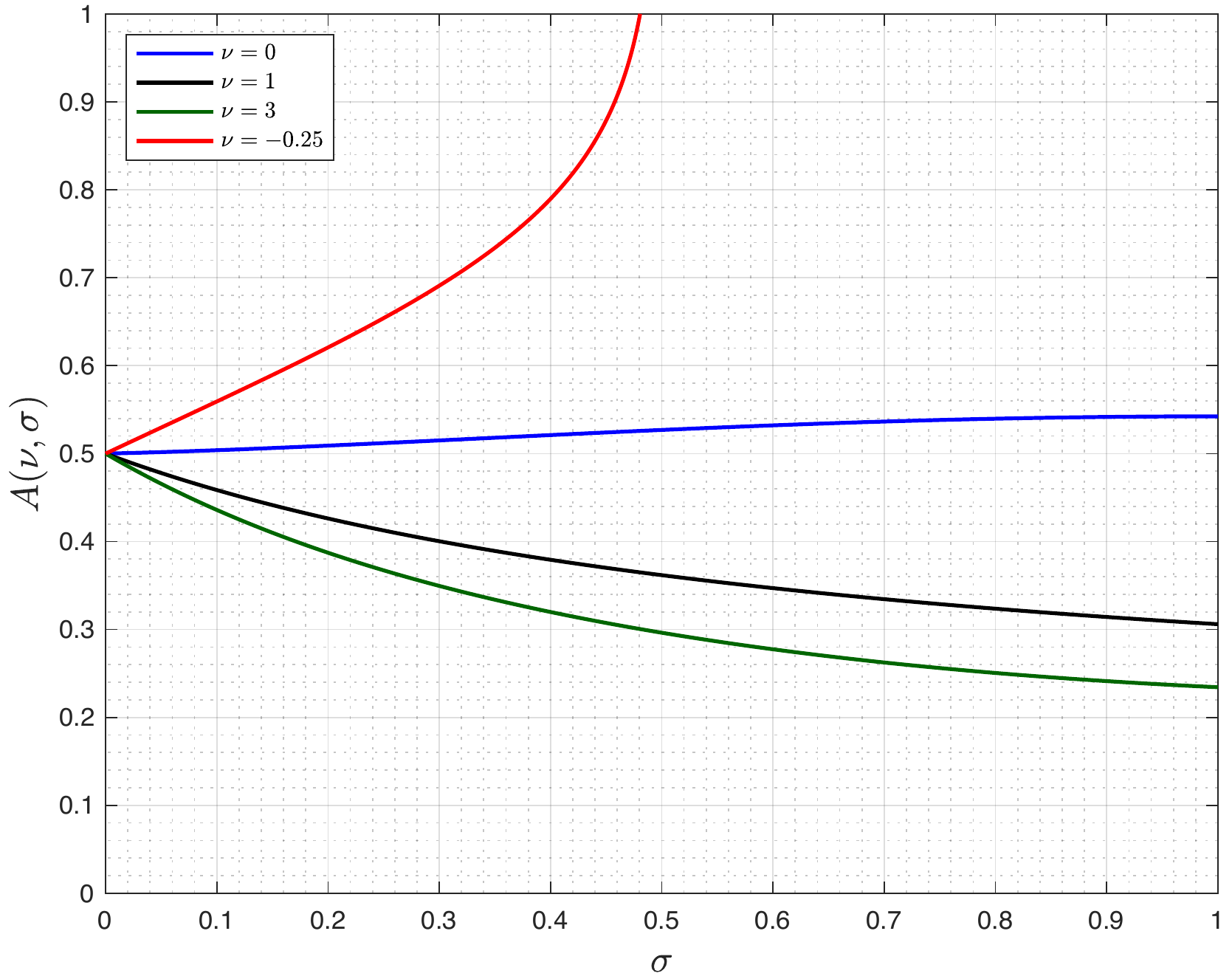}}
\caption{The coefficient $B=B(\nu,\sigma)$ as function of $\sigma\in[0,1]$ for varying $\nu>-\frac{1}{2}$ on the left. On the right we display $A=A(\nu,\sigma)$ as computed in Theorem \ref{res:1} below for $\sigma\in[0,1]$ with varying $\nu>-\frac{1}{2}$ such that $\nu+\frac{1}{2}(1-\sigma)>0$.}
\label{figure1}
\end{figure}
\end{center} 

In addition to \eqref{e:5}, \cite[$4.119$]{MTW} also investigated the behavior of $u(t;\nu,\lambda)$ in the limit $\sigma\uparrow 1$ ($\lambda\uparrow \frac{1}{\pi}$) for small enough $t$ and obtained 
\begin{equation}\label{e:77}
	u\Big(\frac{t}{2};\nu,\frac{1}{\pi}\Big)=\frac{t}{2}\left\{\nu\ln^2t-c(\nu)\ln t+\frac{1}{4\nu}\big(c^2(\nu)-1\big)\right\}+o(1),
\end{equation}
where $c(\nu)=1+2\nu\big(3\ln 2-2\gamma_{E}-\psi_0(1+\nu)\big)$ which is defined in terms  of the digamma function $\psi_0(z)$, cf. \cite[$5.2.2$]{NIST}, and Euler's constant $\gamma_E$, cf. \cite[$5.2.3$]{NIST}. 
\subsection{Hamiltonian structure} Equation \eqref{e:2} is expressible as the Hamiltonian system
\begin{equation}\label{Hami}
	\frac{\d q}{\d t}=\frac{\partial H}{\partial p},\ \ \ \frac{\d p}{\d t}=-\frac{\partial H}{\partial q},\ \ \ \ \ \ \ H=H(q,p,t,\nu)=\frac{t}{2}\sinh^2 q-\frac{p^2}{2t}+4\nu\sinh^2\left(\frac{q}{2}\right)
\end{equation}
with the identification $q=q(t;\nu,\lambda)\equiv\psi(t;\nu,\lambda)$ so that in turn from \eqref{e:1},
\begin{equation}\label{e:7}
	\tau(t;\nu,\lambda)=\exp\left[\frac{1}{2}\int_t^{\infty}H(q,p,s,\nu)\,\d s-\nu\int_t^{\infty}\sinh^2\left(\frac{q}{2}\right)\,\d s\right]\cosh\left(\frac{q}{2}\right).
\end{equation}
This equality identifies $\tau(t;0,\lambda)$ (modulo the $\cosh$ factor) as a tau-function for \eqref{e:2} and the associated transcendent $\psi(t;0,\lambda)$, cf. \cite{JMU}. For $\nu\neq 0$ a similar statement does not appear to be true, however if we allow Painlev\'e-III \eqref{e:4} to re-enter at this point, then \cite[$3.12$]{B} showed that $\tau(t;\nu,\lambda)$ for $\nu>-\frac{1}{2}$ is expressible as a product of two Painlev\'e-III tau-functions. Still, since this interpretation won't play any further role below, we shall simply address \eqref{e:7} as {\it generalized} tau-function for \eqref{e:2} and continue with the discussion of the associated tau-function connection problem:\bigskip

Standard asymptotic techniques based on \eqref{Watson}, \eqref{e:5}, \eqref{e:6} and \eqref{e:77} show that
\begin{equation*}
	\tau(t;\nu,\lambda)\sim1-\frac{\lambda^2}{2}\frac{\Gamma^2(\nu+\frac{1}{2})}{(2t)^{2\nu+1}}\e^{-2t}\left\{\nu-\left(\nu+\frac{1}{2}\right)\left(\nu^2+\frac{3\nu}{2}+1\right)\frac{1}{t}\right\},\  t\rightarrow+\infty,\ \ \lambda\pi\in[0,1],\ \ \nu>-\frac{1}{2},
\end{equation*}
as well as
\begin{equation}\label{prob}
	\tau(t;\nu,\lambda)\sim A(\nu,\lambda)t^{\frac{\sigma}{4}(\sigma-2)},\ \ \ t\downarrow 0,\ \ \lambda\pi\in(0,1],\ \ \ \ \sigma<1+2\nu,
\end{equation}
with some $t$-independent coefficient $A(\nu,\lambda)$. It is straightforward to write down a total integral formula for $A(\nu,\lambda)$ in terms of $\psi(t;\nu,\lambda)$ using \eqref{e:1}, however a simple closed form expression for it cannot be obtained in this elementary way. The explicit computation of $A(\nu,\lambda)$ is known as tau-function connection problem and a special case of it was solved by Tracy \cite{T}, see also Tracy and Widom \cite{TW}. In detail, he computed $A(\nu,\lambda)$ for $\nu=0$,
\begin{equation}\label{Tracy}
	A(0,\lambda)=\e^{3\zeta'(-1)-(3s^2+\frac{1}{6})\ln 2}\big(G(1+s)G(1-s)\big)^{-1},\ \ \ \ s=\frac{1}{2}(1-\sigma)\in\left[0,\frac{1}{2}\right),
\end{equation}
in terms of the Riemann zeta function $\zeta(z)$, see \cite[$25.2.1$]{NIST}, and Barnes G-function $G(z)$, see \cite[$5.17.3$]{NIST}. In this note we prove the following general formula for $A(\nu,\lambda)$.
\begin{theo}\label{res:1} Let $s=\frac{1}{2}(1-\sigma)\in[0,\frac{1}{2}),\nu>-\frac{1}{2}$ where $\sigma=\frac{2}{\pi}\arcsin(\pi\lambda)\in(0,1]$ such that $s+\nu>0$. Then,
\begin{align}
	A(\nu,\lambda)=&\,\,\e^{3\zeta'(-1)-(3s^2+\frac{1}{6})\ln 2}\left(\frac{G^2(1+s)G^2(1-s)}{G(1+s+\nu)G(1-s+\nu)}\right)^{-1}\frac{G^2(\frac{1}{2})\Gamma(\frac{1}{2})}{G^2(\nu+\frac{1}{2})\Gamma(\nu+\frac{1}{2})}\times\nonumber\\
	&\times\exp\left[-\frac{\nu}{2}\ln\left(\frac{\Gamma(1-s+\nu)\Gamma(1+s+\nu)}{\Gamma^2(\nu+\frac{1}{2})}\right)\right](s+\nu)^{\frac{\nu}{2}}\label{e:8}.
\end{align}
\end{theo}
We choose not to simplify the special values $G(\frac{1}{2})$ and $\Gamma(\frac{1}{2})$ in \eqref{e:8} any further 
since \eqref{e:8} now quite obviously degenerates to Tracy's result \eqref{Tracy} for fixed $s\in[0,\frac{1}{2})$ as $\nu\rightarrow0$. The importance of \eqref{Tracy} and its generalization \eqref{e:8} stems from the following application to the scaling hypothesis of  spin-spin functions in the analysis of Wu, McCoy, Tracy and Barouch \cite{BMTW}: As shown by Wu in \cite[$5.7$]{W} the critical correlation $S_N:=\langle\sigma_{00}\sigma_{NN}\rangle_{T=T_c}$ along the diagonal satisfies
\begin{equation}\label{Wu:1}
	S_N\sim G\left(\frac{1}{2}\right)G\left(\frac{3}{2}\right)(1+\tanh^2(\beta_cJ_1))^{\frac{1}{4}}(1-\tanh^2(\beta_cJ_1))^{-\frac{1}{4}}N^{-\frac{1}{4}},\ \ \ N\rightarrow\infty,
\end{equation}
where, compare for instance \eqref{Gspecial} below,
\begin{equation*}
	G\left(\frac{1}{2}\right)G\left(\frac{3}{2}\right)=\e^{3\zeta'(-1)+\frac{1}{12}\ln 2}.
\end{equation*}
In order to prove that the limiting scaling functions $F_{\pm}(t)$ connect to the critical result \eqref{Wu:1} one must then derive the small $t$-expansions of $F_{\pm}(t)$ and confirm that the above numerical constant $G(\frac{1}{2})G(\frac{3}{2})$ is precisely equal to
\begin{equation*}
	\lim_{\substack{T\downarrow\uparrow T_c\\ N\rightarrow\infty}}N^{\frac{1}{4}}\langle\sigma_{00}\sigma_{NN}\rangle\bigg|_{t=0}=2^{-\frac{1}{8}}\lim_{t\downarrow 0}F_{\pm}(t).\footnote{The additional power of two in the right-hand side is needed since on the diagonal $R^{\frac{1}{4}}=2^{\frac{1}{8}}N^{\frac{1}{4}}$.}
\end{equation*}
But this is now an easy task once \eqref{e:8} is available: indeed, from \eqref{scaling}, \eqref{e:77}, \eqref{prob} and \eqref{e:8} we find at once
\begin{equation*}
	2^{-\frac{1}{8}}\lim_{t\downarrow 0}F_{\pm}(t)=2^{\frac{1}{4}}A\left(0,\frac{1}{\pi}\right)=\e^{3\zeta'(-1)+\frac{1}{12}\ln 2},
\end{equation*}
and the connection to \eqref{Wu:1} is therefore rigorously established.

\subsection{Further generalizations} It is also worthwhile to mention that other generalizations of the 2D-Ising $\tau$-function $\tau(t;0,\frac{1}{\pi})$ have been studied in the literature. For instance, in the Jimbo, Miwa, and Sato \cite{SMJ} analysis of holonomic quantum fields one considers instead of \eqref{e:1} the following 
\begin{equation*}
	\widehat{\tau}(t;\theta,\lambda):=\exp\left\{\frac{1}{2}\int_t^{\infty}\left[\sinh^2\phi-\left(\frac{\d\phi}{\d s}\right)^2+\frac{\theta^2}{s^2}\tanh^2\phi\right]\,s\,\d s\right\},\ \ \ \theta\in(-1,1),
\end{equation*}
where $\phi=\phi(t;\theta,\lambda)$ satisfies the differential equation
\begin{equation}\label{PV}
	\frac{\d^2\phi}{\d t^2}+\frac{1}{t}\frac{\d\phi}{\d t}=\frac{1}{2}\sinh(2\phi)+\frac{\theta^2}{t^2}(1-\tanh^2\phi)\tanh\phi,
\end{equation}	
with boundary condition
\begin{equation*}
	\phi(t;\theta,\lambda)\sim 2\lambda K_{\theta}(t),\ \ \ t\rightarrow+\infty,\ \ \lambda\pi\in[0,1],
\end{equation*}
in terms of the modified Bessel function $K_{\theta}(z)$, cf. \cite[$10.25.3$]{NIST}. The ODE \eqref{PV} is a special version of Painlev\'e-V after changing variables and Jimbo \cite{J} subsequently solved part of the $\widehat{\tau}$-function connection problem in 1982, while working on the connection problem for Painlev\'e-V functions. The full solution (which is the analogue of our \eqref{e:8} for the small distance expansion of $\widehat{\tau}$) was given by Basor and Tracy \cite[Theorem $3$]{BT} in 1992 as
\begin{equation*}
	\widehat{\tau}(t;\theta,\lambda)\sim\tau_0(\theta,\lambda)t^{\frac{1}{2}(\sigma^2-\theta^2)},\ \ t\downarrow 0,\ \ \sigma\in[0,1):\ \pi^2\lambda^2=\sin\left(\frac{\pi}{2}(\sigma+\theta)\right)\sin\left(\frac{\pi}{2}(\sigma-\theta)\right),
\end{equation*}
with connection coefficient
\begin{equation*}
	\tau_0(\theta,\lambda)=2^{-2(\alpha^2-\beta^2)}\frac{G(1+\alpha+\beta)G(1+\alpha-\beta)G(1-\alpha+\beta)G(1-\alpha-\beta)}{G(1+2\alpha)G(1-2\alpha)},\ \ 2\alpha=\sigma,\ 2\beta=\theta.
\end{equation*}
Yet another occurrence of $\widehat{\tau}$ can be traced to Federbush's quantum field theory model \cite{F}: there the two point function is expressible in terms of $\widehat{\tau}(t;2\beta,\frac{\im}{\beta}\sin\pi\beta)$, cf. \cite{R}, with coupling constant $\beta$. In the same paper \cite{R}, Ruijsenaars also provided a partial treatment of the connection problem for the Federbush two point function and he derived a series representation for the connection coefficient later on. Still, the series was not summed and the full connection problem only solved by the above mentioned later work of Basor and Tracy \cite{BT}.
\begin{rem} Several other $\tau$-function connection problems in statistical mechanics and field theories were solved in the past. Without going into details, or claiming completeness of the following list, we mention the works of Lenard and Jimbo on impenetrable bosons \cite{L,JMMS}, the analysis of Wu and Wu, McCoy, Tracy, Barouch on Ising correlations \cite{W,BMTW} and the ever growing  random matrix theory themed literature, for instance Widom and Dyson \cite{Wi,D}, Ehrhardt and Krasovsky \cite{E,E2,K}, Deift, Its, Krasovsky \cite{DIK}, Deift, Its, Krasovsky and Zhou \cite{DIKZ} as well as Baik, Buckingham and DiFranco \cite{BBD}.

\end{rem}

\subsection{Methodology and outline of paper} As indicated in the abstract our derivation of \eqref{e:8} will rely on a refinement of the techniques chosen in \cite{B}, that is solely on the Hamiltonian structure \eqref{Hami}, \eqref{e:7} and the known solution of the connection problem for \eqref{e:4}, i.e. the boundary data \eqref{e:5}, \eqref{e:6}, \eqref{e:77}. This is in sharp contrast to all the above mentioned connection problems for $\tau$-functions. There, one relied almost always on a deep connection of the underlying $\tau$-function to the theory of Toeplitz (or Hankel) determinants with possible singular generating functions. In this context powerful operator theoretical or in more recent years Riemann-Hilbert nonlinear steepest descent techniques are readily available in the derivation of asymptotics. Still, these techniques are fairly advanced and their implementation often a technical challenge. The most significant aspect of the present paper is the fact that a quicker and less technical way is available for \eqref{e:1}. In detail we will rewrite the Hamiltonian integral in \eqref{e:7} as action integral plus explicit terms without any integrals. This strategy for asymptotic analysis was first suggested, though not executed, for a generic Sine-Gordon tau-function in \cite[Appendix A.$1$]{IP}. The first successful implementation appeared then in the paper \cite{BIP} on tau-function asymptotics in random matrix theory and in \cite{B} where our \eqref{e:7} was analyzed asymptotically for $\nu=0$. Additionally, the last reference provides in \cite[Section $4$]{B} further discussions on recent advances on Painlev\'e $\tau$-function connection problems, most importantly a short discussion of the important works of Gamayun, Iorgov, Lisovyy, Tykhyy, Its and Prokhorov \cite{GIL,ILT,ILP} (the interested reader is also invited to find more information on Hamiltonian aspects of Painlev\'e tau-functions in \cite{IP1}).\smallskip

Regarding the organization of the remaining sections, we will generalize the above discussed action integral method to $\nu\neq 0$ in \eqref{e:1} when $\tau(t;\nu,\lambda)$ is no longer a classical $\tau$-function for $\psi$. In detail, in Section \ref{sec:2} we rewrite \eqref{e:7} in terms of classical action integrals and $\nu$-derivatives thereof (the last part is the main difference to \cite{B}). After that standard special function manipulations based on \eqref{e:5} and \eqref{e:6} yield our final result \eqref{e:8} in Section \ref{sec:3}.

\section{Proof of Theorem \ref{res:1} - exact identities}\label{sec:2} Our starting point is the following generalization of \cite[$(2.1)$]{B} for $\nu\neq 0$.
\begin{prop} Suppose $q=q(t;\nu,\lambda)$ and $p=p(t;\nu,\lambda)$ solve \eqref{Hami} with boundary condition \eqref{Watson} for any fixed $t>0,\lambda\pi\in[0,1]$ and $\nu>-\frac{1}{2}$. Then
\begin{equation}\label{p:1}
	\int_t^{\infty}H(q,p,s,\nu)\,\d s=-tH(q,p,t,\nu)+S(t;\nu,\lambda)+4\nu\int_t^{\infty}\sinh^2\left(\frac{q}{2}\right)\,\d s,
\end{equation}
where $S=S(t;\nu,\lambda)$ denotes the action integral
\begin{equation*}
	S(t;\nu,\lambda):=\int_t^{\infty}\left(p\frac{\d q}{\d s}-H(q,p,s,\nu)\right)\,\d s.
\end{equation*}
\end{prop}
\begin{proof} Simply $t$-differentiate the right-hand side in \eqref{p:1},
\begin{align*}
	\frac{\d}{\d t}\bigg[-tH(q,p,t)+S(t;\nu,\lambda)&+4\nu\int_t^{\infty}\sinh^2\left(\frac{q}{2}\right)\,\d s\bigg]=-H-t\frac{\partial H}{\partial t}-p\frac{\d q}{\d t}+H-4\nu\sinh^2\left(\frac{q}{2}\right)\\
	&=-\frac{t}{2}\sinh^2q+\frac{p^2}{2t}-4\nu\sinh^2\left(\frac{q}{2}\right)=-H,
\end{align*}
i.e. both sides in \eqref{p:1} have to match, modulo a possible $t$-independent additive term. However, from \eqref{Watson},
\begin{equation*}
	q(t;\nu,\lambda)\sim\frac{2\lambda}{(2t)^{\nu+\frac{1}{2}}}\Gamma\Big(\nu+\frac{1}{2}\Big)\e^{-t},\ \ t\rightarrow+\infty,
\end{equation*}
which implies that both sides in \eqref{p:1} decay exponentially fast at $t=+\infty$. Thus the potential additive term is in fact vanishing and \eqref{p:1} therefore established.
\end{proof}
The presence of the action integral $S(t;\nu,\lambda)$ is the main advantage of identity \eqref{p:1}; namely we can first shift $t$-integration to $\lambda$-integration (this step already appeared in \cite[Corollary $2.3$]{B}).
\begin{cor} For any fixed $t>0,\lambda\pi\in[0,1]$ and $\nu>-\frac{1}{2}$,
\begin{equation}\label{p:2}
	S(t;\nu,\lambda)=-\int_0^{\lambda}p\frac{\partial q}{\partial\lambda'}\,\d\lambda'.
\end{equation}
\end{cor}
\begin{proof} Compare the proof of \cite[Corollary $2.3$]{B}, here subject to \eqref{Watson} and $q(t;\nu,0)\equiv 0$.
\end{proof}
Secondly, the remaining integral $\int_t^{\infty}\sinh^2(\frac{q}{2})\,\d s$ in \eqref{p:1} can also be computed in terms of $S(t;\nu,\lambda)$.
\begin{cor} For any fixed $t>0,\lambda\pi\in[0,1]$ and $\nu>-\frac{1}{2}$,
\begin{equation}\label{p:3}
	\int_t^{\infty}\sinh^2\left(\frac{q}{2}\right)\,\d s=-\frac{1}{4}\left(p\frac{\partial q}{\partial\nu}+\frac{\partial S}{\partial\nu}\right).
\end{equation}
\end{cor}
\begin{proof} We have
\begin{eqnarray*}
	\frac{\partial S}{\partial\nu}&=&\int_t^{\infty}\left[\frac{\partial p}{\partial\nu}\frac{\d q}{\d s}+p\frac{\partial}{\partial\nu}\Big(\frac{\d q}{\d s}\Big)-\frac{\partial H}{\partial q}\frac{\partial q}{\partial\nu}-\frac{\partial H}{\partial p}\frac{\partial p}{\partial\nu}-4\sinh^2\left(\frac{q}{2}\right)\right]\,\d s\\
	&\stackrel{\eqref{Hami}}{=}&\int_t^{\infty}\left[p\frac{\partial}{\partial\nu}\Big(\frac{\d q}{\d s}\Big)+\frac{\d p}{\d s}\frac{\partial q}{\partial\nu}-4\sinh^2\left(\frac{q}{2}\right)\right]\,\d s=-p\frac{\partial q}{\partial\nu}-4\int_t^{\infty}\sinh^2\left(\frac{q}{2}\right)\,\d s,
\end{eqnarray*}
where we integrated by parts in the last equality and used the $\nu$-differentiable asymptotics \eqref{Watson}.
\end{proof}
Merging \eqref{p:1}, \eqref{p:2} and \eqref{p:3} with \eqref{e:7} we arrive at the following exact identity: for any fixed $t>0,\lambda\pi\in[0,1]$ and $\nu>-\frac{1}{2}$,
\begin{equation}\label{p:4}
	\tau(t;\nu,\lambda)=\exp\left[-\frac{t}{2}H(q,p,t,\nu)+\frac{1}{2}S(t;\nu,\lambda)-\frac{\nu}{4}\left(p\frac{\partial q}{\partial\nu}+\frac{\partial S}{\partial\nu}\right)\right]\cosh\Big(\frac{q}{2}\Big),
\end{equation}
where the action $S$ is given by the $\lambda$-integral in \eqref{p:2}. This completes our collection of exact formul\ae\, and we will now derive the small $t$-expansions of $S(t;\nu,\lambda)$ and $H(q,p,t,\nu)$.
\section{Proof of Theorem \ref{res:1} - asymptotic identities}\label{sec:3}
Fix $\lambda\pi\in(0,1),\nu>-\frac{1}{2}$ throughout such that $0<\sigma<1+2\nu$ (this is sufficient since \eqref{prob} holds true for all $\lambda\pi\in(0,1]$ and $\nu>-\frac{1}{2}$ such that $0<\sigma< 1+2\nu$). From \eqref{Hami} and \eqref{e:5}, as $t\downarrow 0$,
\begin{equation*}
	\frac{\partial q}{\partial\lambda}=-\frac{\partial\sigma}{\partial\lambda}\ln t-\frac{1}{B}\frac{\partial B}{\partial\lambda}+\mathcal{O}\big(t^{1-\sigma}\ln t\big);\ \ \ \ \ \ p=\sigma+\mathcal{O}\big(t^{1-\sigma}\big);
\end{equation*}
so that with \eqref{p:2} and $\sigma(0)=0$,
\begin{equation}\label{p:5}
	S(t;\nu,\lambda)=\frac{\sigma^2}{2}\ln t+\int_0^{\lambda}\frac{\sigma(\lambda')}{B(\nu,\lambda')}\frac{\partial B}{\partial\lambda'}(\nu,\lambda')\,\d\lambda'+\mathcal{O}\big(t^{1-\sigma}\big),\ \ t\downarrow 0.
\end{equation}
The last expansion can be further broken down by recalling \eqref{e:6},
\begin{equation}\label{inter}
	S(t;\nu,\lambda)=\frac{\sigma^2}{2}\ln t-\frac{3\sigma^2}{2}\ln 2+J_+(\lambda,\nu)-J_-(\lambda,\nu)-2J_+(\lambda,0)+2J_-(\lambda,0)+\mathcal{O}\big(t^{1-\sigma}\big),
\end{equation}
where
\begin{equation*}
	J_{\pm}(\nu,\lambda):=\int_0^{\lambda}\sigma(\lambda')\frac{\partial}{\partial\lambda'}\ln\Gamma\big(\nu+\frac{1}{2}(1\pm\sigma(\lambda'))\big)\,\d\lambda'.
\end{equation*}
\begin{prop}\label{imp} For any $\lambda\pi\in(0,1),\nu>-\frac{1}{2}$ such that $\sigma<1+2\nu$,
\begin{align*}
	J_+(\nu,\lambda)-J_-(\nu,\lambda)=&\,\frac{\sigma^2}{2}+2\ln\left[\frac{G(1-s+\nu)G(1+s+\nu)}{G^2(\frac{1}{2}+\nu)}\right]+(1-2\nu)\ln\left[\frac{\Gamma(1-s+\nu)\Gamma(1+s+\nu)}{\Gamma^2(\frac{1}{2}+\nu)}\right]\\
	&-2\ln\Gamma(1+s+\nu)+(1+2\nu)\ln(s+\nu),\ \ \ s=\frac{1}{2}(1-\sigma).
\end{align*}
\end{prop}
\begin{proof} Set $s=\frac{1}{2}(1-\sigma)\in(0,\frac{1}{2})$ so that $1-s=\frac{1}{2}(1+\sigma)\in(\frac{1}{2},1)$ and
\begin{equation*}
	J_+(\nu,\lambda)=\int_0^{\lambda}\sigma(\lambda')\frac{\partial}{\partial\lambda'}\ln\Gamma\big(1-s(\lambda')+\nu\big)\,\d\lambda',
\end{equation*}
as well as
\begin{equation*}
	J_-(\nu,\lambda)=\int_0^{\lambda}\sigma(\lambda')\frac{\partial}{\partial\lambda'}\Gamma\big(1+s(\lambda')+\nu\big)\,\d\lambda'-\int_0^{\lambda}\sigma(\lambda')\frac{\partial}{\partial\lambda'}\ln\big(s(\lambda)'+\nu\big)\,\d\lambda'.
\end{equation*}
Integrating by parts, we have thus
\begin{align}
	J_+(\nu,\lambda)-J_-(\nu,\lambda)=\sigma\ln\left\{\frac{\Gamma(1-s+\nu)}{\Gamma(1+s+\nu)}\right\}&-\int_0^{\lambda}\left[\frac{\d\sigma}{\d\lambda'}(\lambda')\right]\ln\left\{\frac{\Gamma(1-s(\lambda')+\nu)}{\Gamma(1+s(\lambda')+\nu)}\right\}\,\d\lambda'\nonumber\\
	&+\sigma\ln(s+\nu)-\int_0^{\lambda}\left[\frac{\d\sigma}{\d\lambda'}(\lambda')\right]\ln\big(s(\lambda')+\nu\big)\,\d\lambda'.\label{p:6}
\end{align}
Next,
\begin{equation}\label{p:7}
	\int_0^{\lambda}\left[\frac{\d\sigma}{\d\lambda'}(\lambda')\right]\ln\big(s(\lambda')+\nu\big)\,\d\lambda'=2\int_{\nu+s}^{\nu+\frac{1}{2}}\ln x\,\d x=2\Big[x\ln x-x\Big]\bigg|_{x=\nu+s}^{\nu+\frac{1}{2}},
\end{equation}
and
\begin{equation}\label{p:8}
	\int_0^{\lambda}\left[\frac{\d\sigma}{\d\lambda'}(\lambda')\right]\ln\left\{\frac{\Gamma(1-s(\lambda')+\nu)}{\Gamma(1+s(\lambda')+\nu)}\right\}\,\d\lambda'=2\int_{\nu-\frac{1}{2}}^{\nu-s}\ln\Gamma(1+x)\,\d x+2\int_{\nu+\frac{1}{2}}^{\nu+s}\ln\Gamma(1+x)\,\d x.
\end{equation}
At this point the Barnes G-function enters our computation via the formula
\begin{equation}\label{Barnes}
	\int_0^z\ln\Gamma(1+x)\,\d x=\frac{z}{2}\ln(2\pi)-\frac{z}{2}(z+1)+z\ln\Gamma(1+z)-\ln G(1+z),\ \ z\in\mathbb{C}:\ \Re z>-1,
\end{equation}
and, combining \eqref{p:6}, \eqref{p:7} with \eqref{p:8} and \eqref{Barnes}, after straightforward simplifications,
\begin{align*}
	J_+(\nu,\lambda)-J_-(\nu,\lambda)=&\,\frac{\sigma^2}{2}+2\ln\left[\frac{G(1-s+\nu)G(1+s+\nu)}{G^2(\frac{1}{2}+\nu)}\right]+(1-2\nu)\ln\left[\frac{\Gamma(1-s+\nu)\Gamma(1+s+\nu)}{\Gamma^2(\frac{1}{2}+\nu)}\right]\\
	&-2\ln\Gamma(1+s+\nu)+(1+2\nu)\ln(s+\nu),
\end{align*}
which proves the Proposition.
\end{proof}
Next, we use Proposition \ref{imp} back in \eqref{inter}, combine it with the well-known special values
\begin{equation}\label{Gspecial}
	\Gamma\left(\frac{1}{2}\right)=\sqrt{\pi},\ \ \ \ 2\ln G\left(\frac{1}{2}\right)=3\zeta'(-1)-\frac{1}{2}\ln\pi+\frac{1}{12}\ln 2
\end{equation}
and arrive at the following small $t$-expansion for the action integral.
\begin{prop}\label{action} For any $\lambda\pi\in(0,1),\nu>-\frac{1}{2}$, as $t\downarrow 0$, with $s+\nu>0$,
\begin{align}
	S(t;\nu,\lambda)=&\frac{\sigma^2}{2}\ln t-\frac{3\sigma^2}{2}\ln 2-\frac{\sigma^2}{2}+6\zeta'(-1)+\frac{1}{6}\ln 2+2\ln\left[\frac{G(1-s+\nu)G(1+s+\nu)}{G^2(1-s)G^2(1+s)}\right]\nonumber\\
	&+2\ln\left[\frac{G^2(\frac{1}{2})\Gamma(\frac{1}{2})}{G^2(\frac{1}{2}+\nu)\Gamma(\frac{1}{2}+\nu)}\right]+\ln\left[\frac{\Gamma^2(s)}{\Gamma^2(1-s)}\frac{\Gamma(1-s+\nu)}{\Gamma(1+s+\nu)}\right]-2\nu\ln\left[\frac{\Gamma(1-s+\nu)\Gamma(1+s+\nu)}{\Gamma^2(\frac{1}{2}+\nu)}\right]\nonumber\\
	&+(1+2\nu)\ln(s+\nu)+\mathcal{O}\big(t^{1-\sigma}\big),\label{p:9}
\end{align}
and the error term is $(t,\nu,\lambda)$-differentiable.
\end{prop}
Note that all logarithms above are well-defined and real-valued since $\Gamma(x)>0,G(x)>0$ for $x>0$ and by our assumption $s+\nu>0$. After the action integral, the small $t$-behavior of the Hamiltonian is much easier, see \eqref{Hami} and \eqref{e:5},
\begin{prop}\label{Hamiltonian} For any $\lambda\pi\in(0,1),\nu>-\frac{1}{2}$, as $t\downarrow 0$, with $s+\nu>0$,
\begin{equation*}
	H(q,p,t,\nu)=-\frac{\sigma^2}{2t}+\mathcal{O}\big(t^{-\sigma}\big).
\end{equation*}
\end{prop}
Next, we merge \eqref{p:3}, \eqref{p:9} and compute the outstanding integral $\int_t^{\infty}\sinh^2(\frac{q}{2})\,\d s$ in the $t\downarrow 0$ limit.
\begin{prop}\label{deriv} For any $\lambda\pi\in(0,1),\nu>-\frac{1}{2}$, as $t\downarrow 0$, with $s+\nu>0$,
\begin{equation}\label{p:10}
	p\frac{\partial q}{\partial\nu}+\frac{\partial S}{\partial\nu}=-2\ln\left[\frac{\Gamma(1-s+\nu)\Gamma(1+s+\nu)}{\Gamma^2(\frac{1}{2}+\nu)}\right]+2\ln(s+\nu)+\mathcal{O}\big(t^{1-\sigma}\big).
\end{equation}
\end{prop}
\begin{proof} By definition, the digamma function \cite[$5.2.2$]{NIST} equals $\psi_0(z):=\Gamma'(z)/\Gamma(z)$ for $z\notin\mathbb{Z}_{\leq 0}$. Moreover
\begin{equation*}
	G'(z)=G(z)\left[(z-1)\psi_0(z)-z+\frac{1}{2}\ln(2\pi)+\frac{1}{2}\right],\ \ z\notin\mathbb{Z}_{\leq 0}
\end{equation*}
which follows from \eqref{Barnes} by differentiation. Thus, $\nu$-differentiating \eqref{p:9}, we find
\begin{align}\label{p:11}
	\frac{\partial S}{\partial\nu}=\sigma\left[\psi_0\left(\nu+\frac{1}{2}(1+\sigma)\right)-\psi_0\left(\nu+\frac{1}{2}(1-\sigma)\right)\right]&\,-2\ln\left[\frac{\Gamma(1-s+\nu)\Gamma(1+s+\nu)}{\Gamma^2(\frac{1}{2}+\nu)}\right]+2\ln(s+\nu)\nonumber\\
	&+\mathcal{O}\big(t^{1-\sigma}\big),\ \ \ t\downarrow 0.
\end{align}
Since also
\begin{equation}\label{p:12}
	p\frac{\partial q}{\partial\nu}=-\frac{\sigma}{B}\frac{\partial B}{\partial\nu}+\mathcal{O}\big(t^{1-\sigma}\big)=-\sigma\left[\psi_0\left(\nu+\frac{1}{2}(1+\sigma)\right)-\psi_0\left(\nu+\frac{1}{2}(1-\sigma)\right)\right]+\mathcal{O}\big(t^{1-\sigma}\big),\ \ t\downarrow 0,
\end{equation}
expansion \eqref{p:10} follows now from \eqref{p:11} and \eqref{p:12}.
\end{proof}
Towards the end of our derivation we are now left with combining our result: first  by Propositions \ref{action}, \ref{Hamiltonian} and \ref{deriv},
\begin{align*}
	\exp\bigg[-&\frac{t}{2}H(q,p,t,\nu)+\frac{1}{2}S(t;\nu,\lambda)-\frac{\nu}{4}\left(p\frac{\partial q}{\partial\nu}+\frac{\partial S}{\partial\nu}\right)\bigg]=t^{\frac{1}{4}\sigma^2}\e^{3\zeta'(-1)-\frac{3}{4}\sigma^2\ln 2+\frac{1}{12}\ln 2}\times\\
	&\times\,\left(\frac{G^2(1-s)G^2(1+s)}{G(1-s+\nu)G(1+s+\nu)}\right)^{-1}\frac{G^2(\frac{1}{2})\Gamma(\frac{1}{2})}{G^2(\frac{1}{2}+\nu)\Gamma(\frac{1}{2}+\nu)}\exp\left[-\frac{\nu}{2}\ln\left(\frac{\Gamma(1-s+\nu)\Gamma(1+s+\nu)}{\Gamma^2(\nu+\frac{1}{2})}\right)\right]\times\\
	&\times(s+\nu)^{\frac{\nu}{2}}\exp\left[\frac{1}{2}\ln\left(\frac{\Gamma^2(s)}{\Gamma^2(1-s)}\frac{\Gamma(1-s+\nu)(s+\nu)}{\Gamma(1+s+\nu)}\right)\right]\left(1+\mathcal{O}\big(t^{1-\sigma}\big)\right),\ \ t\downarrow 0,\ \ s+\nu>0.
\end{align*}
On the other hand, see \eqref{e:5}, \eqref{e:6}, as $t\downarrow 0$ with $s+\nu>0$,
\begin{equation*}
	\cosh\left(\frac{q}{2}\right)=t^{-\frac{1}{2}\sigma}\e^{(\frac{3}{2}\sigma-1)\ln 2}\exp\left[-\frac{1}{2}\ln\left(\frac{\Gamma^2(s)}{\Gamma^2(1-s)}\frac{\Gamma(1-s+\nu)}{\Gamma(s+\nu)}\right)\right]\left(1+\mathcal{O}\big(\max\{t^{1-\sigma},t^{\sigma}\}\big)\right)
\end{equation*}
so that all together, 
\begin{align*}
	\tau(t;\nu,\lambda)=&\,\e^{3\zeta'(-1)-(\frac{3}{4}\sigma^2-\frac{3}{2}\sigma+\frac{11}{12})\ln 2}\left(\frac{G^2(1-s)G^2(1+s)}{G(1-s+\nu)G(1+s+\nu)}\right)^{-1}\frac{G^2(\frac{1}{2})\Gamma(\frac{1}{2})}{G^2(\frac{1}{2}+\nu)\Gamma(\frac{1}{2}+\nu)}\times\\
	&\times\exp\left[-\frac{\nu}{2}\ln\left(\frac{\Gamma(1-s+\nu)\Gamma(1+s+\nu)}{\Gamma^2(\nu+\frac{1}{2})}\right)\right](s+\nu)^{\frac{\nu}{2}}\,t^{\frac{\sigma}{4}(\sigma-2)}\left(1+\mathcal{O}\big(\max\{t^{1-\sigma},t^{\sigma}\}\big)\right),
\end{align*}
as $t\downarrow 0$, uniformly in $(s,\nu)\in(0,\frac{1}{2})\times(-\frac{1}{2},+\infty)$ chosen from compact subsets such that $s+\nu>0$. The last expansion matches to leading order precisely \eqref{prob} and thus completes the proof of Theorem \ref{res:1}.

\end{document}